\documentclass[aps,amsmath,amssymb,letter,scriptaddress,twocolumn, prl,showkeys]{revtex4}

	\usepackage{amsmath}
	\usepackage{amsthm}
	\newtheorem{thm}{Theorem}
	\usepackage{amssymb}
	\usepackage{pgfplots, alphalph}
	\usepgfplotslibrary{groupplots}
	\usepackage{mathtools}
	\usepackage{makeidx}
	\usepackage{amsfonts}
	\usepackage[ansinew]{inputenc}
	\usepackage[usenames,dvipsnames]{pstricks}
	\usepackage{subfigure}
	\usepackage{epsfig}
	\usepackage{pst-grad} 
	\usepackage{pst-plot} 
	\usepackage[colorlinks,hyperindex]{hyperref}

\makeatletter

\newcommand{\Rmnum}[1]{\expandafter\@slowromancap\romannumeral #1@}
\makeatother

\makeatletter
\pgfplotsset{
auto title/.style={title=(\AlphAlph{\pgfplots@group@current@plot})
    }
}
\makeatother

\makeindex

\begin{document}

\title{The mathematics of the genetic code reveal that frequency recurrence leads to nonlinear scaling in the DNA codon distribution of {\it Homo sapiens}}

\author{Bohdan B. Khomtchouk}
\email{b.khomtchouk@med.miami.edu}
\affiliation{University of Miami Miller School of Medicine \\ Center for Therapeutic Innovation and Department of Psychiatry and Behavioral Sciences \\ 1120 NW 14th ST Suite 1463, Miami, FL, USA 33136}

\begin{abstract}
  The nature of the quantitative distribution of the 64 DNA codons in the human genome has been an issue of debate for over a decade. Some groups have proposed that the quantitative distribution of the DNA codons ordered as a rank-frequency plot follows a well-known power law called Zipf's law. Others have shown that the DNA codon distribution is best fitted to an exponential function. However, the reason for such scaling behavior has not yet been addressed. In the present study, we demonstrate that the nonlinearity of the DNA codon distribution is a direct consequence of the frequency recurrence of the codon usage (i.e., the repetitiveness of codon usage frequencies at the whole genome level).  We discover that if frequency recurrence is absent from the human genome, the frequency of occurrence of codons scales linearly with the codon rank. We also show that DNA codons of both low and high frequency of occurrence in the genome are best fitted by an exponential function and provide strong evidence to suggest that the coding region of the human genome does not follow Zipf's law. Information-theoretic methods and entropy calculations are applied to the DNA codon distribution and a new approach, called the lariat method, is proposed to quantitatively analyze the DNA codon distribution in {\it Homo sapiens}.
\end{abstract}


\maketitle

\section{Introduction}

From the days of its inception, information theory has served to explain how communicative systems function. The mathematician Claude E. Shannon launched the field of information theory with his seminal 1948 paper, {\it A Mathematical Theory of Communication} \cite{Shannon1}, which was later published as a book {\it The Mathematical Theory of Communication} \cite{Shannonbook}. The theory laid the groundwork for the development of data compression and storage methods (e.g., ZIP files, JPEGs, MP3s), channel coding (e.g., DSL), national security measures (e.g., cryptographically secure ciphers), and various other commercial applications (e.g., seismic oil exploration). The applications of information theory have also expanded into various academic fields such as quantum computing, neurobiology, linguistics, and ecology.

The application of information theory to the study of biological phenomena dates back to the 1970s \cite{Glatin, Reichert, Guiasu}, where it was used to obtain quantitative measures such as redundancy and divergence of DNA sequences \cite{Roman}. A resurgence of the application of information-theoretic methods to the study of DNA sequences was prompted by the availability of sequence data in publicly available online databases from the late 1980s to the late 1990s and early 2000s \cite{Altschul:1991aa, Li, Peng, Voss, Grosse2000, Grossesymposium}.  More recently, information theory has reappeared in the genetics and genomics field in the study of alignment-free DNA sequence analysis and comparison, genome entropy estimation, and the identification of allergens in sequenced genomes \cite{Vinga:2013aa, Dang}. The resurgence of the use of information-theoretic methods in genetics and genomics is predicted to infuse promising results into next-generation sequencing projects and gene mapping, metagenomics, and communication theory-based models of information transmission in organisms \cite{Vinga:2013aa}.

In this article, we build a communication theory-based model of the DNA genetic code as a communicative system, where the speaker is modeled as the 64 codons and the receiver is modeled as the 20 amino acids and a stop signal.  In information theory terminology, the codons are the signals and the amino acids and stop signal are the objects.  The 64 signals comprise the entirety of the lexicon and the 21 objects are the targets to be mapped into from this list of signals through the communicative channel, the RNA intermediate.

\section{Background and Methods}

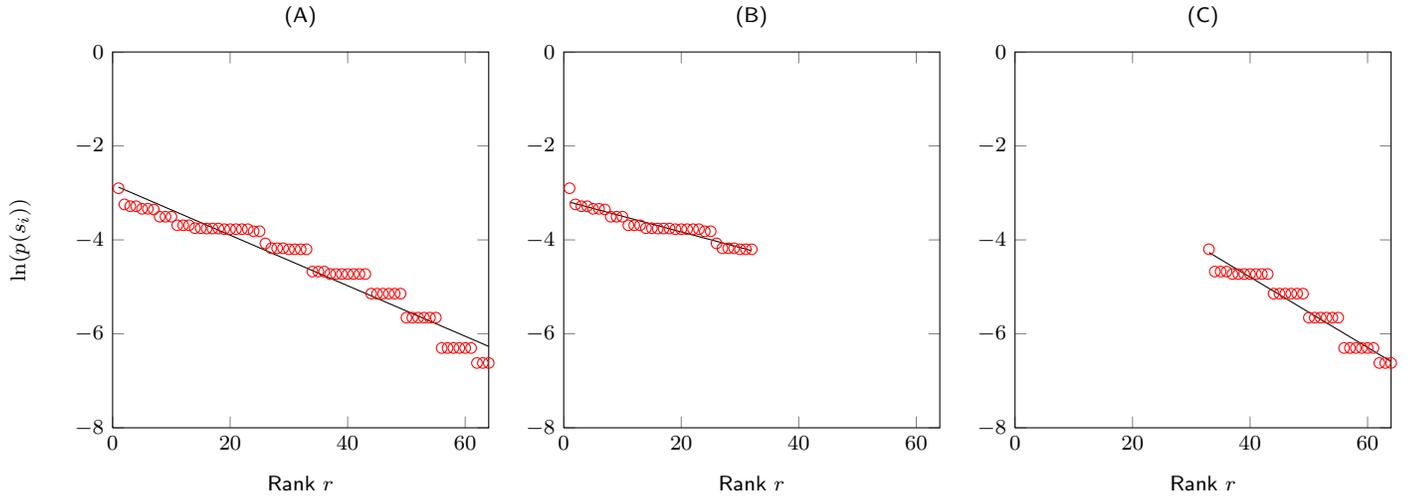
\begin{figure*}
    \begin{tikzpicture}[font=\footnotesize\sffamily]
      \begin{groupplot}[
         group style={group size=3 by 1,
	ylabels at=edge left
	},
          view={0}{90},
          width=5cm,
          height=5cm,
          scale only axis,
          xmin=0, xmax=64,
          ymin=-8, ymax=0,
	xlabel={Rank $r$},
    	ylabel={$\ln(p(s_{i}))$},
          name=plot2,
          unbounded coords=jump]
        ]
        \nextgroupplot [auto title]       
\addplot [only marks, mark=o, color=red] table{
1	-2.900422094
2	-3.244193633
3	-3.283414346
4	-3.283414346
5	-3.338222583
6	-3.338222583
7	-3.352407217
8	-3.506557897
9	-3.506557897
10	-3.506557897
11	-3.688879454
12	-3.688879454
13	-3.688879454
14	-3.750754858
15	-3.750754858
16	-3.757872326
17	-3.757872326
18	-3.757872326
19	-3.772261063
20	-3.772261063
21	-3.772261063
22	-3.772261063
23	-3.772261063
24	-3.816712826
25	-3.816712826
26	-4.074541935
27	-4.177726171
28	-4.177726171
29	-4.177726171
30	-4.199705078
31	-4.199705078
32	-4.199705078
33	-4.199705078
34	-4.674163057
35	-4.674163057
36	-4.674163057
37	-4.726531043
38	-4.726531043
39	-4.726531043
40	-4.726531043
41	-4.726531043
42	-4.726531043
43	-4.726531043
44	-5.144166687
45	-5.144166687
46	-5.144166687
47	-5.144166687
48	-5.144166687
49	-5.144166687
50	-5.65499231
51	-5.65499231
52	-5.65499231
53	-5.65499231
54	-5.65499231
55	-5.65499231
56	-6.301619475
57	-6.301619475
58	-6.301619475
59	-6.301619475
60	-6.301619475
61	-6.301619475
62	-6.620073207
63	-6.620073207
64	-6.620073207
};
\addplot [domain=1:64, color=black] {-0.0538*x-2.8234};
        \nextgroupplot  [auto title]       
\addplot [only marks, mark=o, color=red] table{
1	-2.900422094
2	-3.244193633
3	-3.283414346
4	-3.283414346
5	-3.338222583
6	-3.338222583
7	-3.352407217
8	-3.506557897
9	-3.506557897
10	-3.506557897
11	-3.688879454
12	-3.688879454
13	-3.688879454
14	-3.750754858
15	-3.750754858
16	-3.757872326
17	-3.757872326
18	-3.757872326
19	-3.772261063
20	-3.772261063
21	-3.772261063
22	-3.772261063
23	-3.772261063
24	-3.816712826
25	-3.816712826
26	-4.074541935
27	-4.177726171
28	-4.177726171
29	-4.177726171
30	-4.199705078
31	-4.199705078
32	-4.199705078
};
\addplot [domain=1:32, color=black] {-0.0333*x-3.1637};
        \nextgroupplot  [auto title]       
\addplot [only marks, mark=o, color=red] table{
33	-4.199705078
34	-4.674163057
35	-4.674163057
36	-4.674163057
37	-4.726531043
38	-4.726531043
39	-4.726531043
40	-4.726531043
41	-4.726531043
42	-4.726531043
43	-4.726531043
44	-5.144166687
45	-5.144166687
46	-5.144166687
47	-5.144166687
48	-5.144166687
49	-5.144166687
50	-5.65499231
51	-5.65499231
52	-5.65499231
53	-5.65499231
54	-5.65499231
55	-5.65499231
56	-6.301619475
57	-6.301619475
58	-6.301619475
59	-6.301619475
60	-6.301619475
61	-6.301619475
62	-6.620073207
63	-6.620073207
64	-6.620073207
};
\addplot [domain=33:64, color=black] {-0.0749*x-1.7959};
 \end{groupplot}
    \end{tikzpicture}
\caption{Panel A: Linear regression analysis of exponential fit data shows that linear regression analysis over the entire set (64 codons) of the lexicon of DNA codons demonstrates the superiority of exponential scaling ($R^{2}$=0.9454) in fitting DNA sequences as compared to power law scaling ($R^{2}$=0.6823, Panel A: Fig.\,\ref{fig:linregpower}).  Panel B: Linear regression analysis over a subset of the lexicon size of DNA codons (r=$1,\ldots,32$) shows the superiority of exponential fits at data of low rank ($R^{2}$=0.9074) over power law fits to the respective data ($R^{2}$=0.8707, Panel B: Fig.\,\ref{fig:linregpower}).  Panel C: Linear regression analysis over a subset of the lexicon size of DNA codons (r=$33,\ldots,64$) also shows the superiority of exponential fits at data of high rank ($R^{2}$=0.9403) over power law fits to the respective data ($R^{2}$=0.9114, Panel C: Fig.\,\ref{fig:linregpower}).} 
\label{fig:linregexp}
  \end{figure*}

There has been a considerable range of analytical approaches \cite{Vinga:2013aa, Peng, Peng1994, Borodovsky, Buldyrev, Ossadnik, Viswanathan, Viswanathan1997, Buldyrevbook, Buldyrev1993, Buldyrev1995, Mantegna, Mantegna1995, Audit, Azad, Stanley, Arneodo, Allegrini, Stanleybook, Azbel, Havlin, Bernaola, Bernaola1996, Lu, Herzel, Frappat, Som} performed to investigate the statistical and scale invariant features of long-range correlations in DNA to gain insight into questions such as whether or not the rank-frequency distribution of the codons follows a power law scaling law known as Zipf's law \cite{Zipf}:

\begin{equation} 
f \propto \frac{1}{r^{\alpha}} 
\label{eq:zipf} 
\end{equation}
where $f$ is the frequency, $r$ is the rank, and $\alpha$ is a statistical scaling coefficient that is classically seen to be $\approx 1$ for many sources examined, such as texts \cite{Ferrer2005}, where a text can be a string of DNA nucleotides.

We analyze {\it Homo sapiens} data from the Codon Usage Database and known amino acid residue frequencies sampled from the primary structures of 207 unrelated proteins of known sequence \cite{Nakamura, Klapper}. 

Let $p(s_{i})$ be the probability of use of a specific codon $s_{i}$ ($i$=$1,\ldots,64$), generally for any amino acid or stop signal $r_{j}$ ($j$=$1,\ldots,21$):

\begin{equation}
 \label{eq:test}
 p(s_{i})=\sum_{j} p(s_{i},r_{j}) 
\end{equation}
where $p(s_{i},r_{j})$ is the probability of using $s_{i}$ for $r_{j}$ \cite{Ferrer2003}:

\begin{equation}
 \label{eq:test}
 p(s_{i},r_{j})=p(r_{j})p(s_{i}|r_{j})
\end{equation}
That is, $p(s_{i},r_{j})$ is the probability of $s_{i}$ mapping onto $r_{j}$, $p(r_{j})$ is the probability of use of amino acid $r_{j}$ in protein sequences, and $p(s_{i}|r_{j})$ is the probability of using codon  for amino acid $r_{j}$.  

The frequency of occurrence of a codon can be expressed in terms of a probability:

\begin{equation}
 \label{eq:psi}
 p(s_{i})=\frac{L(s_{i})}{L} 
\end{equation}
where $L(s_{i})$ is the frequency of occurrence of the specific codon $s_{i}$, and $L$ is the total frequency of occurrence of all the 64 codons in the lexicon.  Therefore:

\begin{equation}
 \label{eq:test}
\sum_{i=1}^{64} p(s_{i})=1
\end{equation}

In the linear regression analysis, least squares fitting is applied using a slope no-intercept model to log-transformed data to examine the respective power law and exponential fits to the DNA codon distribution, where the fits are conducted with a one parameter model using least squares.

\section{Results and Discussion}

We graph $p(s_{i})$ versus the rank $r$ and examine the nature of both a power law fit and an exponential fit to the data (Fig.\,\ref{fig:data}).

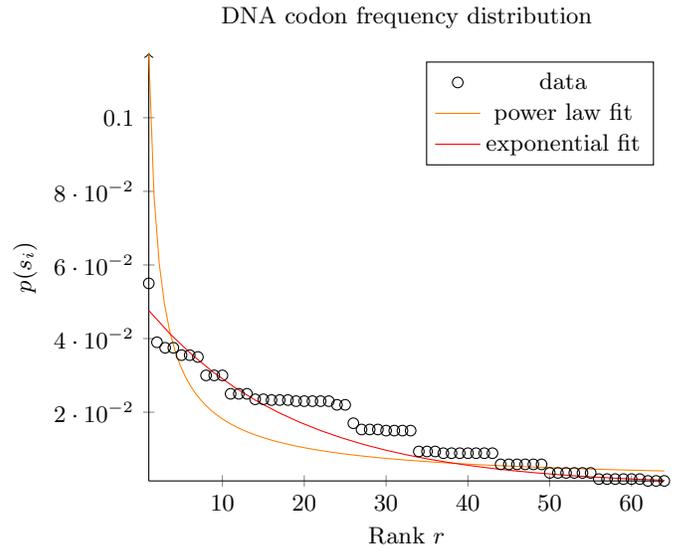
\begin{figure}[htbp]
	\begin{center}
		\begin{tikzpicture}
\begin{axis}[ 
    title={DNA codon frequency distribution},
    axis lines=middle,
    axis line style={->},
    ylabel near ticks,
    xlabel near ticks,
    xlabel={Rank $r$},
    ylabel={$p(s_{i})$}]

\addplot [only marks, mark=o] table{
1	5.50E-02
2	3.90E-02
3	3.75E-02
4	3.75E-02
5	3.55E-02
6	3.55E-02
7	3.50E-02
8	3.00E-02
9	3.00E-02
10	3.00E-02
11	2.50E-02
12	2.50E-02
13	2.50E-02
14	2.35E-02
15	2.35E-02
16	2.33E-02
17	2.33E-02
18	2.33E-02
19	2.30E-02
20	2.30E-02
21	2.30E-02
22	2.30E-02
23	2.30E-02
24	2.20E-02
25	2.20E-02
26	1.70E-02
27	1.53E-02
28	1.53E-02
29	1.53E-02
30	1.50E-02
31	1.50E-02
32	1.50E-02
33	1.50E-02
34	9.33E-03
35	9.33E-03
36	9.33E-03
37	8.86E-03
38	8.86E-03
39	8.86E-03
40	8.86E-03
41	8.86E-03
42	8.86E-03
43	8.86E-03
44	5.83E-03
45	5.83E-03
46	5.83E-03
47	5.83E-03
48	5.83E-03
49	5.83E-03
50	3.50E-03
51	3.50E-03
52	3.50E-03
53	3.50E-03
54	3.50E-03
55	3.50E-03
56	1.83E-03
57	1.83E-03
58	1.83E-03
59	1.83E-03
60	1.83E-03
61	1.83E-03
62	1.33E-03
63	1.33E-03
64	1.33E-03
};

\addplot [domain=1:64, samples=100, color=orange] {0.1176*(x^-0.811)};
\addplot [domain=1:64, color=red] {0.0503*(e^-0.055*x)};
\legend{data,power law fit, exponential fit}

\end{axis} 
\end{tikzpicture}
		\caption{Probability of use of a specific codon {$s_i$} as a function of rank demonstrates frequency of codon usage across the entire distribution spectrum of the lexicon $L$, comprised of the 64 DNA codons.  Power law ($R^{2}$=0.6859) and exponential distributions ($R^{2}$=0.9489) are fitted to the data. Exponential fit: $p(s_{i})=0.0503e^{-0.055r}$ Power law fit: $p(s_{i})=0.1176r^{-0.811}$}
		\label{fig:data}
	\end{center}
\end{figure}

  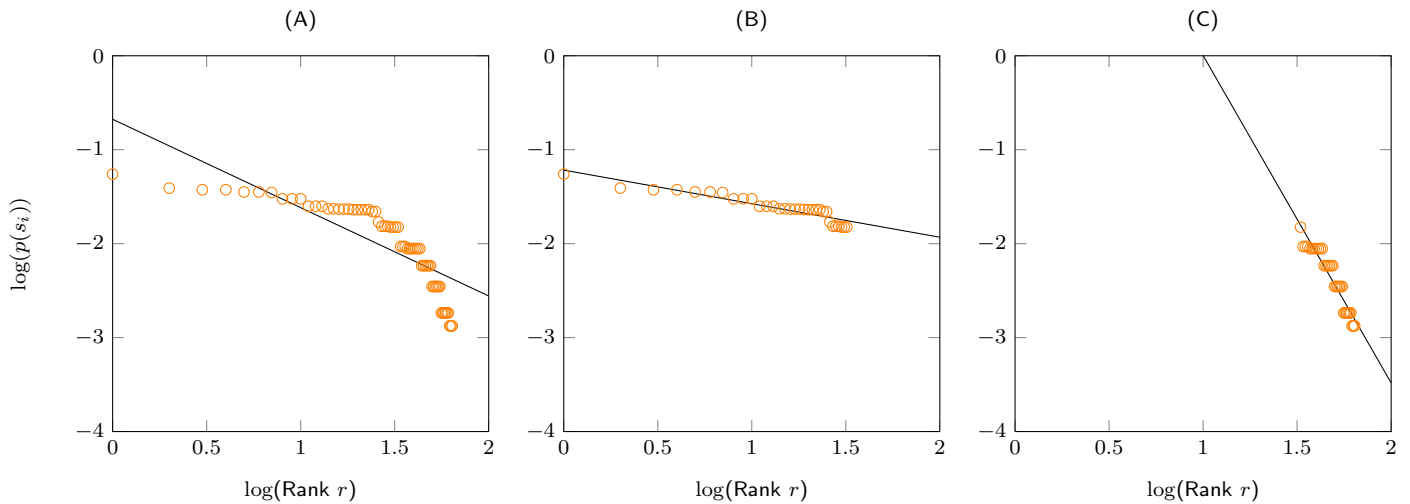
\begin{figure*}
    \begin{tikzpicture}[font=\footnotesize\sffamily]
      \begin{groupplot}[
         group style={group size=3 by 1,
	ylabels at=edge left
	},
          view={0}{90},
          width=5cm,
          height=5cm,
          scale only axis,
          xmin=0, xmax=2,
          ymin=-4, ymax=0,
	xlabel={$\log$(Rank $r$)},
    	ylabel={$\log(p(s_{i}))$},
          name=plot2,
          unbounded coords=jump]
        ]
        \nextgroupplot [auto title]        
\addplot [only marks, mark=o, color=orange] table{
0			-1.259637311
0.301029996	-1.408935393
0.477121255	-1.425968732
0.602059991	-1.425968732
0.698970004	-1.449771647
0.77815125	-1.449771647
0.84509804	-1.455931956
0.903089987	-1.522878745
0.954242509	-1.522878745
1			-1.522878745
1.041392685	-1.602059991
1.079181246	-1.602059991
1.113943352	-1.602059991
1.146128036	-1.628932138
1.176091259	-1.628932138
1.204119983	-1.632023215
1.230448921	-1.632023215
1.255272505	-1.632023215
1.278753601	-1.638272164
1.301029996	-1.638272164
1.322219295	-1.638272164
1.342422681	-1.638272164
1.361727836	-1.638272164
1.380211242	-1.657577319
1.397940009	-1.657577319
1.414973348	-1.769551079
1.431363764	-1.814363423
1.447158031	-1.814363423
1.462397998	-1.814363423
1.477121255	-1.823908741
1.491361694	-1.823908741
1.505149978	-1.823908741
1.51851394	-1.823908741
1.531478917	-2.029963223
1.544068044	-2.029963223
1.556302501	-2.029963223
1.568201724	-2.052706351
1.579783597	-2.052706351
1.591064607	-2.052706351
1.602059991	-2.052706351
1.612783857	-2.052706351
1.62324929	-2.052706351
1.633468456	-2.052706351
1.643452676	-2.234083206
1.653212514	-2.234083206
1.662757832	-2.234083206
1.672097858	-2.234083206
1.681241237	-2.234083206
1.69019608	-2.234083206
1.698970004	-2.455931956
1.707570176	-2.455931956
1.716003344	-2.455931956
1.72427587	-2.455931956
1.73239376	-2.455931956
1.740362689	-2.455931956
1.748188027	-2.736758565
1.755874856	-2.736758565
1.763427994	-2.736758565
1.770852012	-2.736758565
1.77815125	-2.736758565
1.785329835	-2.736758565
1.792391689	-2.875061263
1.799340549	-2.875061263
1.806179974	-2.875061263
};
\addplot [domain=0:2, color=black] {-0.9404*x-0.6761};
        \nextgroupplot [auto title]        
\addplot [only marks, mark=o, color=orange] table{
0			-1.259637311
0.301029996	-1.408935393
0.477121255	-1.425968732
0.602059991	-1.425968732
0.698970004	-1.449771647
0.77815125	-1.449771647
0.84509804	-1.455931956
0.903089987	-1.522878745
0.954242509	-1.522878745
1			-1.522878745
1.041392685	-1.602059991
1.079181246	-1.602059991
1.113943352	-1.602059991
1.146128036	-1.628932138
1.176091259	-1.628932138
1.204119983	-1.632023215
1.230448921	-1.632023215
1.255272505	-1.632023215
1.278753601	-1.638272164
1.301029996	-1.638272164
1.322219295	-1.638272164
1.342422681	-1.638272164
1.361727836	-1.638272164
1.380211242	-1.657577319
1.397940009	-1.657577319
1.414973348	-1.769551079
1.431363764	-1.814363423
1.447158031	-1.814363423
1.462397998	-1.814363423
1.477121255	-1.823908741
1.491361694	-1.823908741
1.505149978	-1.823908741
};
\addplot [domain=0:2, color=black] {-0.3572*x-1.217};
        \nextgroupplot [auto title]        
\addplot [only marks, mark=o, color=orange] table{
1.51851394	-1.823908741
1.531478917	-2.029963223
1.544068044	-2.029963223
1.556302501	-2.029963223
1.568201724	-2.052706351
1.579783597	-2.052706351
1.591064607	-2.052706351
1.602059991	-2.052706351
1.612783857	-2.052706351
1.62324929	-2.052706351
1.633468456	-2.052706351
1.643452676	-2.234083206
1.653212514	-2.234083206
1.662757832	-2.234083206
1.672097858	-2.234083206
1.681241237	-2.234083206
1.69019608	-2.234083206
1.698970004	-2.455931956
1.707570176	-2.455931956
1.716003344	-2.455931956
1.72427587	-2.455931956
1.73239376	-2.455931956
1.740362689	-2.455931956
1.748188027	-2.736758565
1.755874856	-2.736758565
1.763427994	-2.736758565
1.770852012	-2.736758565
1.77815125	-2.736758565
1.785329835	-2.736758565
1.792391689	-2.875061263
1.799340549	-2.875061263
1.806179974	-2.875061263
};
\addplot [domain=0:2, color=black] {-3.485*x+3.4881};
 \end{groupplot}
    \end{tikzpicture}
\caption{Panel A: Linear regression analysis of power law data shows that linear regression analysis over the entire set (64 codons) of the lexicon of DNA codons demonstrates the superiority of exponential scaling ($R^{2}$=0.9454, Panel A: Fig.\,\ref{fig:linregexp}) in fitting DNA sequences as compared to power law scaling ($R^{2}$=0.6823).  Panel B: Linear regression analysis over a subset of the lexicon size of DNA codons (r=$1,\ldots,32$) shows the superiority of exponential fits at data of low rank ($R^{2}$=0.9074, Panel B: Fig.\,\ref{fig:linregexp}) over power law fits to the respective data ($R^{2}$=0.8707).  Panel C: Linear regression analysis over a subset of the lexicon size of DNA codons (r=$33,\ldots,64$) also shows the superiority of exponential fits at data of high rank ($R^{2}$=0.9403, Panel C: Fig.\,\ref{fig:linregexp}) over power law fits to the respective data ($R^{2}$=0.9114).} 
 \label{fig:linregpower}
  \end{figure*}

We discover that the value of the scaling coefficient \eqref{eq:zipf} is $\alpha\approx 0.8$, in significant contrast to Zipf's law where $\alpha\approx 1$ for texts and non-technical natural languages \cite{Ferrer2005, Czirok}, and the power law fit ($R^{2}=0.6859$) is considerably weaker than the exponential fit ($R^{2}=0.9489$).  Therefore, these results strongly suggest that Zipf's law is not applicable to the coding regions of the human genome. The findings support \cite{Som, Frappat, Tsonis} that exponential scaling captures the nature of the DNA codon distribution much more accurately than power law scaling, as evidenced in this analysis by the coefficient of determination, $R^{2}$.  A prior study \cite{Frappat} also demonstrated that the DNA codon frequency distribution as a function of rank {\it r} of {\it Homo sapiens} (Fig.\,\ref{fig:data}), as well as other eukaryotic species, is best modeled as the sum of an exponential function, a linear function, and a constant, not Zipf's law, i.e.:

\begin{equation} 
f(r) = {\alpha}e^{-{\eta}r}-{\beta}r+{\gamma}
\label{eq:zipf} 
\end{equation}

Where $\alpha$, $\eta$, $\beta$, and $\gamma$ are scaling factors that are constant depending on the biological species under analysis.
Furthermore, we extend our analysis to examine high-ranked codons and low-ranked codons with a linear regression approach to determine whether an exponential fit over the respective region of the DNA codon distribution provides a better fit for DNA sequences.  We find that exponential fits (Fig.\,\ref{fig:linregexp}) provide a better representation of DNA codons of lower rank ($r$=$1,\ldots,32$) as well as DNA codons of higher rank ($r$=$33,\ldots,64$) than do the respective power law fits (Fig.\,\ref{fig:linregpower}), as shown by the higher respective $R^{2}$ values.  Hence, we validate that exponential scaling behavior best governs the distribution of the lexicon size of the DNA codons across both low and high ranks. We also conclude from the coefficients of determination obtained in the linear regression analysis that the exponential fit over the entire DNA codon distribution is better equipped to capture the global topology of the frequency distribution than is the power law fit (Fig.\,\ref{fig:linregexp}, Fig.\,\ref{fig:data}, Fig.\,\ref{fig:linregpower}).

Next, we visualize the lexicon as a function of rank (low versus high) to create an intuitive handle for the number of signals (codons) of a certain rank that are present within an entire lexicon (e.g., how many codons of a specific frequency exist in the lexicon, where the lexicon is defined as the 64 DNA codons). This new approach we introduce, which we call the lariat method, poses a natural improvement over existing methodologies to quantify rank based on the frequency when one is interested at examining the distribution of all the different frequencies of the DNA codons at the whole-genome level.  Most prior studies were designed such that DNA codons of recurrent frequency are assigned different rank numbers, where if two codons have the same frequency of occurrence they belong to two different ranks, one following the other sequentially \cite{Som}.  Another study investigating a physical phenomenon in a different academic field wholly unrelated to DNA also assigned sequential ranks to recurrent frequencies \cite{Martinez}. The alternative and physically more meaningful scenario to examine a rank-frequency distribution is to gather the codons into bins of different recurrent frequencies and then assign rank numbers to each bin, where one bin may contain multiple codons.  This procedure partitions the lexicon size of the 64 DNA codons into non-overlapping subsets of signals, or codons, of recurrent frequency:   

\begin{equation}
 \label{eq:test}
 \left\{\\L\right\}=\bigcup_{r} L(r) 
\end{equation}

It has been shown that power-law fits and exponential fits that connect frequency to rank as applied to the codon distribution have been promising sources of fit to DNA sequences \cite{Voss, Som, Tsonis}. Observations that exponential fits have consistently lower $\chi^{2}$ \cite{Som} than power-law fits and, hence, provide a better fit \cite{Tsonis} for DNA sequences have opened interesting new questions about the nature of rank-frequency distributions and the parameters that govern them.  

By employing the lariat method, where codons of recurrent frequency are binned into a single rank value (Fig.\,\ref{fig:lasso}), we discover that the recurrence of the codon usage explains the nonlinearity of the spatial distribution of the DNA codons.  If the recurrence is omitted, as performed in this binning procedure, $p(s_{i})$ scales linearly with the rank $r$.  We verify from the residuals plot (Fig.\,\ref{fig:residuals}) that the linear fit residuals are distributed randomly about zero, signifying that, apart from random uncertainty, the linear model correctly predicts the data. On the contrary, the exponential and particularly the power law fit residuals show systematic, non-random deviation of the data from the respective models.  It should be noted that the number of bins in the lariat method is uniquely determined based upon how many unique codon frequencies exist in the codon pool of the organism under study.  For example, as an extreme hypothetical case of an organism that possesses only three codons, two of which occur at the same frequency, the lariat method would establish two bins: one bin (containing one codon) for the unique codon frequency, and one bin (containing two codons) for the other codon frequency.  As such, the lariat method can be easily extended from studies in humans to other biological organisms.

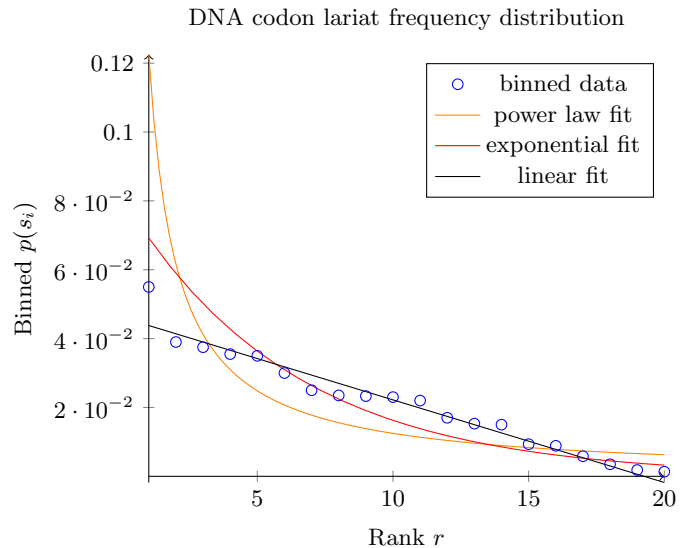
\begin{figure}[htbp]
	\begin{center}
		\begin{tikzpicture}
\begin{axis}[ 
    title={DNA codon lariat frequency distribution},
    axis lines=middle,
    axis line style={->},
    ylabel near ticks,
    xlabel near ticks,
    xlabel={Rank $r$},
    ylabel={Binned $p(s_{i})$}]

\addplot [only marks, mark=o, color=blue] table{
1	5.50E-02
2	3.90E-02
3	3.75E-02
4	3.55E-02
5	3.50E-02
6	3.00E-02
7	2.50E-02
8	2.35E-02
9	2.33E-02
10	2.30E-02
11	2.20E-02
12	1.70E-02
13	1.53E-02
14	1.50E-02
15	9.33E-03
16	8.86E-03
17	5.83E-03
18	3.50E-03
19	1.83E-03
20	1.33E-03
};

\addplot [domain=1:20, samples=100, color=orange] {0.1225*(x^-0.992)};
\addplot [domain=1:20, color=red] {0.0813*(e^-0.161*x)};
\addplot [domain=1:20, color=black] {-0.0024*x+0.0462};
\legend{binned data,power law fit, exponential fit, linear fit}

\end{axis} 
\end{tikzpicture}
		\caption{If the frequency recurrence of codons is omitted, the frequency of occurrence of codons scales linearly with the rank.  Any nonlinear deviations from linear scaling are the result of recurrent codon frequencies.  Power law ($R^{2}$=0.6122), exponential ($R^{2}$=0.8537), and linear fits ($R^{2}$=0.9506) are applied to this binned data. Exponential fit: $p(s_{i})=0.0813e^{-0.161r}$ Linear fit: $p(s_{i})=-0.0024r+0.0462$ Power law fit: $p(s_{i})=0.1225r^{-0.992}$} 
		\label{fig:lasso}
	\end{center}
\end{figure}

The linearity of the data becomes even more pronounced if the first-ranked codon is omitted.  This phenomenon raises interesting questions regarding the biological utility of nonlinear scaling in the genetic code.  The data suggests that the degree of nonlinear, exponential scaling observed in a DNA codon distribution is directly determined by the degree of recurrence of the frequency of the codon usage within the codon lexicon, for any species, not just {\it H. sapiens}.  The evolutionary implications of this are rather interesting, considering that a divergence away from a linear rank-frequency DNA codon distribution allows for less frequent, or less popular, codons (i.e., high rank) in the genome to occur with close to the same frequency as more popular codons (i.e., low rank).  Examining the potential evolutionary pressure that may have driven the genetic code to exhibit this kind of scaling phenomena is a subject of future research. 

\begin{figure}[htbp]
	\begin{center}
		\begin{tikzpicture}
\begin{axis}[ 
    title={Residuals plot of lariat data},
    axis lines=middle,
    axis line style={->},
    legend pos=south east,
    ylabel near ticks,
    xlabel near ticks,
    xlabel={Rank $r$},
    ylabel={Residual}
]

\addplot [only marks, mark=o, color=orange] table{
1	-6.75E-02
2	-2.26E-02
3	-3.69E-03
4	4.53E-03
5	1.02E-02
6	9.29E-03
7	7.23E-03
8	7.93E-03
9	9.48E-03
10	1.05E-02
11	1.06E-02
12	6.59E-03
13	5.71E-03
14	6.06E-03
15	9.88E-04
16	1.03E-03
17	-1.54E-03
18	-3.46E-03
19	-4.77E-03
20	-4.94E-03
};

\addplot [only marks, mark=o, color=red] table{
1	-1.42E-02
2	-1.99E-02
3	-1.27E-02
4	-7.20E-03
5	-1.35E-03
6	-9.43E-04
7	-1.34E-03
8	1.08E-03
9	4.24E-03
10	6.75E-03
11	8.17E-03
12	5.22E-03
13	5.31E-03
14	6.47E-03
15	2.07E-03
16	2.67E-03
17	5.68E-04
18	-9.82E-04
19	-1.98E-03
20	-1.92E-03
};

\addplot [only marks, mark=o, color=black] table{
1	1.12E-02
2	-2.40E-03
3	-1.50E-03
4	-1.10E-03
5	8.00E-04
6	-1.80E-03
7	-4.40E-03
8	-3.50E-03
9	-1.27E-03
10	8.00E-04
11	2.20E-03
12	-4.00E-04
13	3.33E-04
14	2.40E-03
15	-8.67E-04
16	1.06E-03
17	4.33E-04
18	5.00E-04
19	1.23E-03
20	3.13E-03
};

\legend{power law residuals, exponential residuals, linear residuals}

\end{axis} 
\end{tikzpicture}
		\caption{Linear fit residuals are randomly distributed about zero showing no systematic, non-random deviation of the original dataset (Fig.\,\ref{fig:data}) from the linear fit model, in contrast to the exponential and power law models.}
		\label{fig:residuals}	
	\end{center}	
\end{figure}
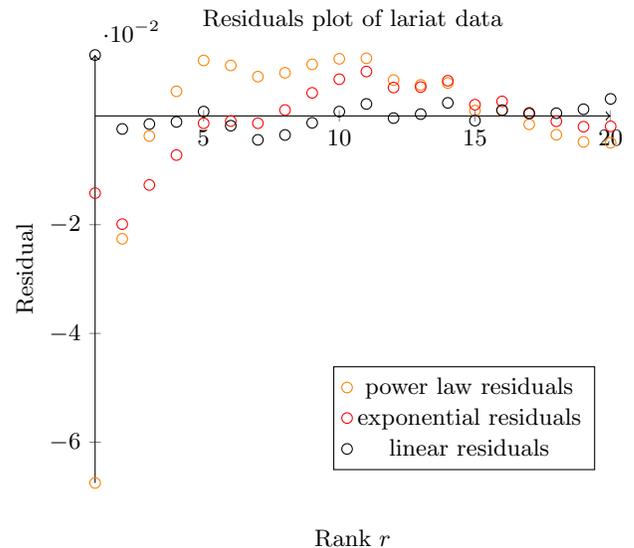

We indeed see from the original dataset (Fig.\,\ref{fig:data}) that the frequency recurrence of the codon usage causes the shape of the codon distribution to assume a nonlinear form, where our linear regression approach has revealed strong linearity in rank-frequency dependency of the DNA codons when their frequency recurrence is controlled for by the lariat method (Fig.\,\ref{fig:lasso}).  As such, a high-ranked codon does not occur with a much lower frequency than a low-ranked codon, as would occur if the data (the DNA codon rank-frequency distribution) was scaled linearly.  Therefore, the exponential scaling seems to serve a certain buffering capacity, the biological significance of which is now an open question.  We subsequently followed up this result with an entropy calculation on the original dataset (Fig.\,\ref{fig:data}) to establish a quantitative measure evaluating the tendency of the 64 codons to roughly evenly distribute across the genome with respect to frequency:

\begin{equation}
 \label{eq:ent}
 H(S)=-\sum_{i=1}^{N} p(s_{i})\log_N(p(s_{i}))=0.926
\end{equation}
In a hypothetical scenario where all the DNA codons are evenly distributed with respect to frequency regardless of the rank (i.e., all the codons occur with the same frequency), the entropy, $H(S)$, of such a system is unity.  A rigorous mathematical proof of this result is demonstrated in Appendix; however, this result can also be arrived to using Lagrange multipliers.  It is a well-known fact that the unique probability distribution having maximum entropy is the uniform distribution; therefore, the Appendix serves merely to demonstrate a new proof of a famous information-theoretic result as applied to the human genome.  As the hypothetical scenario of a uniformly distributed codon frequency is clearly not applicable to the human genome, an entropy value of this magnitude ($\approx 1$) suggests the presence of biological mechanisms that ensure that even though certain codons are more prevalent than others in quantity due to frequency recurrence, the spatial distribution of the DNA codons in the genome behaves as if to mask this effect.

\section{Conclusion}

We provide new information-theoretic analyses which strongly suggest that the coding region of the human genome does not behave according to Zipf's law.  We prove that if the 64 DNA codons of the human genetic code are not equiprobable, then the entropy, $H(S)$, of the genetic code is less than unity.  We also discover that any deviation away from linear rank-frequency DNA codon scaling is a consequence of the recurrence of the frequency of the codon usage.  Hence, we show that the reason for the existence of exponential scaling in the human genome is the repetitiveness in the frequency of usage of different codons.  We show that if frequency recurrence were to be absent from the human genome, the frequency of occurrence of codons would scale linearly with the codon rank, instead of exponentially.  This linear vs. exponential scaling dichotomy creates interesting open questions regarding the potential evolutionary driving force which has led to the rise of the preference of exponential scaling in the DNA genetic code, as opposed to other kinds of scaling (e.g., linear scaling or Zipfian scaling).  We leave the potential evolutionary implications of our findings as a subject of future research to experts in the field.

\section{Acknowledgments}
BBK wishes to acknowledge the support of the Department of Defense (DoD) through the National Defense Science \& Engineering Graduate Fellowship (NDSEG) Program.  BBK wishes to thank Wolfgang Nonner for useful discussions and scientific guidance, and Claes Wahlestedt, Georges St. Laurent \Rmnum{3}, Seth J. Schwartz, and Hemant Ishwaran for critical review and helpful comments on the manuscript.

\newpage
\begin{widetext}
\appendix
\section{Appendix A: Mathematical Proofs} 
\label{App:proof}

\begin{thm}
\label{thm:unity}
If the 64 DNA codons of the human genetic code are equiprobable, then the entropy, $H(S)$, of the genetic code is unity.
\end{thm}

\begin{proof}
Let each of the 64 DNA codons that comprise the human genetic code to occur with the same frequency.  Then it follows from Eq.\eqref{eq:psi} that, for all codons $s_{i}$ in the genetic code, $p(s_{i})$=$1/N$ where $N=64$.  From Eq.\eqref{eq:ent} it follows that:
\begin{equation}
 \label{eq:thm1}
 H(S)=-\sum_{i=1}^{N} p(s_{i})\log_N(p(s_{i}))=-\sum_{i=1}^{N} p(s_{i})\frac{\ln(p(s_{i}))}{\ln(N)}=-\sum_{i=1}^{N} \frac{1}{N}\frac{\ln(1/N)}{\ln(N)}=-\frac{1}{N}\sum_{i=1}^{N} \frac{-\ln(N)}{\ln(N)}=1
\end{equation}    
\end{proof}

\begin{thm}
\label{thm:notunity}
If the 64 DNA codons of the human genetic code are not equiprobable, then the entropy, $H(S)$, of the genetic code is less than unity.
\end{thm}

\begin{proof}
Suppose there exists one DNA codon in the human genetic code that does not occur with the same frequency as the other 63 codons.  Let $\epsilon>0$.  Then for some given $s_{i}$, it follows from Theorem(\ref{thm:unity}) that $p(s_{i})$=$1/N+\epsilon$ and:
\begin{equation}
\label{eq:thm2}
\begin{split}
H(S)=-\sum_{i=1}^{N} p(s_{i})\frac{\ln(p(s_{i}))}{\ln(N)}=-\Biggl(\frac{(1/N + \epsilon)\ln(1/N + \epsilon)}{\ln(N)}+\frac{(1/N - \epsilon)\ln(1/N - \epsilon)}{\ln(N)}+\frac{(1/N)\ln(1/N)}{\ln(N)}+\cdots \\
+\cdots+\frac{(1/N)\ln(1/N)}{\ln(N)}\Biggr)
\end{split}
\end{equation}
where if $p(s_{i})$ changes by $+\epsilon$ for one codon, it is necessarily true that $p(s_{i})$ will change by $-\epsilon$ for some other codon.  

Considering the polarity of the entropy definition, to prove this theorem it must be demonstrated from \eqref{eq:thm2} that:

\begin{equation}
\label{eq:inequality}
H(S)<-\frac{64(1/N)(\ln(1/N))}{\ln(N)}
\end{equation}

To show \eqref{eq:inequality} we proceed directly from \eqref{eq:thm2}:

\begin{equation}
\label{multi}
\begin{split}
H(S)=-\Biggl(&\frac{(1/N + \epsilon)[\ln(1/N) + \ln(1+N\epsilon)]}{\ln(N)}+\frac{(1/N - \epsilon)[\ln(1/N) + \ln(1-N\epsilon)]}{\ln(N)}+\frac{(1/N)\ln(1/N)}{\ln(N)}+\cdots \\ 
&+\cdots+\frac{(1/N)\ln(1/N)}{\ln(N)}\Biggr) \overset{?}<-\frac{64(1/N)(\ln(1/N))}{\ln(N)}
\end{split}
\end{equation}

Subtracting the $p(s_{i})=1/N$ terms from both sides:

\begin{equation}
\label{multi2}
\begin{split}
-\Biggl(&\frac{(1/N + \epsilon)[\ln(1/N) + \ln(1+N\epsilon)]}{\ln(N)}+\frac{(1/N - \epsilon)[\ln(1/N) + \ln(1-N\epsilon)]}{\ln(N)}\Biggr) \overset{?}<-\frac{2(1/N)(\ln(1/N))}{\ln(N)}
\end{split}
\end{equation}

Cancelling the polarities and the $\ln(N)$ terms on both sides of the equation:

\begin{equation}
\label{multi2}
\begin{split}
\Bigl({(1/N + \epsilon)[\ln(1/N) + \ln(1+N\epsilon)]}+{(1/N - \epsilon)[\ln(1/N) + \ln(1-N\epsilon)]}\Bigr) \overset{?}<{2(1/N)(\ln(1/N))}
\end{split}
\end{equation}

Expanding out the equation, combining like terms, and cancelling on both sides leads to:

\begin{equation}
\label{multi3}
\begin{split}
\frac{1}{N}{\ln(1+N\epsilon)}+\epsilon\ln(1+N\epsilon)+\frac{1}{N}{\ln(1-N\epsilon)}-\epsilon\ln(1-N\epsilon) \overset{?}<0
\end{split}
\end{equation}

Regrouping terms leads to:

\begin{equation}
\label{multi4}
\begin{split}
\ln(1+N\epsilon)+\ln(1-N\epsilon)+N\epsilon(\ln(1+N\epsilon)-\ln(1-N\epsilon)) \overset{?}<0
\end{split}
\end{equation}

Since $\epsilon>0$ is arbitrarily small, it follows that:

\begin{equation}
\label{multi5}
\begin{split}
\ln(1+N\epsilon)+\ln(1-N\epsilon) \overset{?}<0
\end{split}
\end{equation}

Employing the series expansion definition of $\ln(1+x)$:

\begin{equation}
\label{lnseries}
\begin{split}
\ln(1+x)=x-\frac{1}{2}x^{2}+\frac{1}{3}x^{3}-\frac{1}{4}x^{4}+\cdots for -1<x<1
\end{split}
\end{equation}

Subsequent substitution and algebra yields:

\begin{equation}
\label{series}
\begin{split}
\ln(1+N\epsilon)+\ln(1-N\epsilon)=-(N\epsilon)^{2}-\frac{1}{2}(N\epsilon)^{4}-\frac{1}{3}(N\epsilon)^{6}-\frac{1}{4}(N\epsilon)^{8}-\cdots-\frac{(N\epsilon)^{2k}}{k}
\end{split}
\end{equation}

Hence it has been proven that for all $n \geq k$:

\begin{equation}
\label{multi5}
\begin{split}
\ln(1+N\epsilon)+\ln(1-N\epsilon)=-\sum_{k=1}^{n}\frac{(N\epsilon)^{2k}}{k}<0 
\end{split}
\end{equation}

Therefore we have proven \eqref{eq:inequality} which proceeds directly from \eqref{eq:thm2}.  This completes the proof.

\end{proof}
\end{widetext}

\bibliographystyle{plain}
\bibliography{Khomtchouk_bib}

\begin{thebibliography}{10}

\bibitem{Shannon1}
C.E. Shannon.
\newblock A mathematical theory of communication.
\newblock {\em Bell System Technical Journal}, 27, 1948.

\bibitem{Shannonbook}
C.E. Shannon.
\newblock {\em The Mathematical Theory of Communication}.
\newblock University of Illinois Press, 1949.

\bibitem{Glatin}
{\em Information Theory and the Living System}.
\newblock Columbia University Press, 1972.

\bibitem{Reichert}
T.A. Reichert, D.N. Cohen, and A.K.C. Wong.
\newblock An application of information theory to genetic mutations and the
  matching of polypeptide chains.
\newblock {\em Journal of Theoretical Biology}, 1973.

\bibitem{Guiasu}
S.~Guiasu.
\newblock {\em Information Theory With Applications}.
\newblock McGraw-Hill, 1977.

\bibitem{Roman}
R.~Roman-Rold\'{a}n, P.~Bernaola-Galv\'{a}n, and J.L. Oliver.
\newblock Application of information theory to {DNA} sequence analysis: a
  review.
\newblock {\em Pattern Recognition}, 29(7), 1996.

\bibitem{Altschul:1991aa}
S.F. Altschul.
\newblock Amino acid substitution matrices from an information theoretic
  perspective.
\newblock {\em J Mol Biol}, 219(3):555--65, Jun 1991.

\bibitem{Li}
W.~Li and K.~Kaneko.
\newblock Long-range correlation and partial $1/{f^{\alpha}}$ spectrum in a
  noncoding {DNA} sequence.
\newblock {\em Europhysics Letters}, 17(7), 1992.

\bibitem{Peng}
C.K.~Peng et. al.
\newblock Long-range correlations in nucleotide sequences.
\newblock {\em Nature}, 356, 1992.

\bibitem{Voss}
R.F. Voss.
\newblock Evolution of long-range fractal correlations and ${1/f}$ noise in
  {DNA} base sequences.
\newblock {\em Physical Review Letters}, 68(25), 1992.

\bibitem{Grosse2000}
I.~Grosse et~al.
\newblock Species independence of mutual information in coding and noncoding
  dna.
\newblock {\em Physical Review E}, 61(5), 2000.

\bibitem{Grossesymposium}
I.~Grosse et~al.
\newblock Average mutual information of coding and noncoding dna.
\newblock In R.B. Altman, A.K. Dunker, L.~Hunter, K.~Lauderdale, and T.E.
  Klein, editors, {\em Pacific Symposium on Biocomputing 2000}, 2000.

\bibitem{Vinga:2013aa}
S.~Vinga.
\newblock Information theory applications for biological sequence analysis.
\newblock {\em Brief Bioinform}, Sep 2013.

\bibitem{Dang}
H.X. Dang and C.B. Lawrence.
\newblock Allerdictor: fast allergen prediction using text classification
  techniques.
\newblock {\em Bioinformatics}, 2014.

\bibitem{Peng1994}
C.K.~Peng et. al.
\newblock Mosaic organization of dna nucleotides.
\newblock {\em Physical Review E}, 49(2), 1994.

\bibitem{Borodovsky}
M.~Borodovsky and S.M. Gusein-Zade.
\newblock A general rule for ranged series of codon frequencies in different
  genomes.
\newblock {\em Journal of Biomolecular Structure and Dynamics}, 6(5), 1989.

\bibitem{Buldyrev}
S.V.~Buldyrev et~al.
\newblock Long-range fractal correlations in {DNA}.
\newblock {\em Physical Review Letters}, 71, 1993.

\bibitem{Ossadnik}
S.M~Ossadnik et~al.
\newblock Correlation approach to identify coding regions in {DNA} sequences.
\newblock {\em Biophysical Journal}, 67, 1994.

\bibitem{Viswanathan}
G.M.~Viswanathan et~al.
\newblock Long-range correlation measures for quantifying patchiness:
  Deviations from uniform power-law scaling in genomic {DNA}.
\newblock {\em Physica A}, 249, 1998.

\bibitem{Viswanathan1997}
G.M.~Viswanathan et~al.
\newblock Quantification of {DNA} patchiness using correlation measures.
\newblock {\em Biophysical Journal}, 72, 1997.

\bibitem{Buldyrevbook}
S.V. Buldyrev.
\newblock {\em Power Laws, Scale-Free Networks, and Genome Biology}.
\newblock Springer Science+Business Media, 2006.

\bibitem{Buldyrev1993}
S.V.~Buldyrev et~al.
\newblock On long-range power law correlations in {DNA}.
\newblock {\em Physical Review Letters}, 71, 1993.

\bibitem{Buldyrev1995}
S.V.~Buldyrev et~al.
\newblock Long-range correlation properties of coding and noncoding dna
  sequences: Genbank analysis.
\newblock {\em Physical Review Letters E}, 51(5), 1995.

\bibitem{Mantegna}
R.N.~Mantegna et~al.
\newblock Linguistic features of noncoding dna sequences.
\newblock {\em Physical Review Letters}, 73(23), 1994.

\bibitem{Mantegna1995}
R.N.~Mantegna et~al.
\newblock Systematic analysis of coding and noncoding {DNA} sequences using
  methods of statistical linguistics.
\newblock {\em Physical Review E}, 52(3), 1995.

\bibitem{Audit}
B.~Audit et~al.
\newblock Long-range correlations in genomic {DNA}: A signature of the
  nucleosomal structure.
\newblock {\em Physical Review Letters}, 86(11), 2001.

\bibitem{Azad}
R.K.~Azad et~al.
\newblock Segmentation of genomic {DNA} through entropic divergence: Power laws
  and scaling.
\newblock {\em Physical Review Letters E}, 65, 2002.

\bibitem{Stanley}
H.E.~Stanley et~al.
\newblock Scaling features of noncoding {DNA}.
\newblock {\em Physica A}, 273, 1999.

\bibitem{Arneodo}
A.~Arneodo et~al.
\newblock Characterizing long-range correlations in {DNA} sequences from
  wavelet analysis.
\newblock {\em Physical Review Letters}, 74(16), 1995.

\bibitem{Allegrini}
P.~Allegrini et~al.
\newblock Dynamical model for {DNA} sequences.
\newblock {\em Physical Review E}, 52(5), 1995.

\bibitem{Stanleybook}
H.E.~Stanley et~al.
\newblock {\em Scale Invariant Features of Coding and Noncoding {DNA}
  Sequences}.
\newblock Fractal Geometry in Biological Systems: An Analytical Approach. CRC
  Press, 1996.

\bibitem{Azbel}
M.Y. Azbel.
\newblock Universality in a {DNA} statistical structure.
\newblock {\em Physical Review Letters}, 75(1), 1995.

\bibitem{Havlin}
S.~Havlin et~al.
\newblock Statistical and linguistic properties of {DNA} sequences.
\newblock {\em Fractals}, 3, 1995.

\bibitem{Bernaola}
P.~Bernaola-Galv\'{a}n et~al.
\newblock Finding borders between coding and noncoding {DNA} regions by an
  entropic segmentation method.
\newblock {\em Physical Review Letters}, 85(6), 2000.

\bibitem{Bernaola1996}
P.~Bernaola-Galv\'{a}n, R.~Roman-Rold\'{a}n, and J.L. Oliver.
\newblock Compositional segmentation and long-range fractal correlations in
  {DNA} sequences.
\newblock {\em Physical Review E}, 53(5), 1996.

\bibitem{Lu}
X.~Lu et~al.
\newblock Characterizing self-similarity in bacteria {DNA} sequences.
\newblock {\em Physical Review E}, 58(3), 1998.

\bibitem{Herzel}
H.~Herzel and I.~Gro{\ss}e.
\newblock Correlations in {DNA} sequences: The role of protein coding segments.
\newblock {\em Physical Review E}, 55(1), 1997.

\bibitem{Frappat}
L.~Frappat et~al.
\newblock Universality and {S}hannon entropy of codon usage.
\newblock {\em Physical Review E}, 68, 2003.

\bibitem{Som}
A.~Som et~al.
\newblock Codon distributions in {DNA}.
\newblock {\em Physical Review Letters E}, 63, 2001.

\bibitem{Zipf}
G.~Zipf.
\newblock {\em Human Behaviour and the Principle of Least Effort: An
  Introduction to Human Ecology}.
\newblock Addison-Wesley, 1949.

\bibitem{Ferrer2005}
R.~Ferrer i~Cancho.
\newblock Zipf's law from a communicative phase transition.
\newblock {\em European Physical Journal B}, 47, 2005.

\bibitem{Nakamura}
Y.~Nakamura, T.~Gojobori, and T.~Ikemura.
\newblock Codon usage tabulated from international {DNA} sequence databases:
  status for the year 2000.
\newblock {\em Nucleic Acids Research}, 28(1), 2000.

\bibitem{Klapper}
M.H. Klapper.
\newblock The independent distribution of amino acid near neighbor pairs into
  polypeptides.
\newblock {\em Biochemical and Biophysical Research Communications}, 78(3),
  1977.

\bibitem{Ferrer2003}
R.~Ferrer i~Cancho.
\newblock Least efforts and the origins of scaling in human language.
\newblock {\em Proceedings of the National Academy of Sciences}, 100(3), 2003.

\bibitem{Czirok}
A.~Czir\'{o}k et~al.
\newblock Correlations in binary sequences and a generalized {Z}ipf analysis.
\newblock {\em Physical Review E}, 52, 1995.

\bibitem{Tsonis}
A.A. Tsonis, J.B. Elsner, and P.A. Tsonis.
\newblock Is {DNA} a language?
\newblock {\em Journal of Theoretical Biology}, 184, 1997.

\bibitem{Martinez}
G.~Mart\'{i}nez-Mekler et~al.
\newblock Universality of rank-ordering distributions in the arts and sciences.
\newblock {\em PLoS One}, 4(3), 2009.

\end{thebibliography}
\bibliographystyle{unsrt}

\end{document}